\documentclass[twocolumn,prl,showpacs,superscriptaddress]{revtex4-1}
\usepackage[latin9]{inputenc}
\usepackage{amsmath}
\usepackage{amssymb}
\usepackage{amsthm}
\usepackage{graphicx}
\usepackage{color}

\theoremstyle{plain}
\newtheorem*{theorem*}{Theorem}

\makeatletter

\@ifundefined{textcolor}{}
{%
 \definecolor{BLACK}{gray}{0}
 \definecolor{WHITE}{gray}{1}
 \definecolor{RED}{rgb}{1,0,0}
 \definecolor{GREEN}{rgb}{0,1,0}
 \definecolor{BLUE}{rgb}{0,0,1}
 \definecolor{CYAN}{cmyk}{1,0,0,0}
 \definecolor{MAGENTA}{cmyk}{0,1,0,0}
 \definecolor{YELLOW}{cmyk}{0,0,1,0}
}

\usepackage{color}
\usepackage{amsfonts}
\usepackage{graphics}
\usepackage{epsfig}\usepackage{amsthm}

\def\identity{\leavevmode\hbox{\small1\kern-3.8pt\normalsize1}}

\newcommand{\ket}[1]{\left | #1 \right\rangle}
\newcommand{\bra}[1]{\left \langle #1 \right |}

\renewcommand{\epsilon}{\varepsilon}

\makeatother

\begin{document}

\title{Revealing non-classicality of inaccessible objects}

\author{Tanjung Krisnanda}
\affiliation{School of Physical and Mathematical Sciences, Nanyang Technological University, Singapore}

\author{Margherita Zuppardo}
\affiliation{School of Physical and Mathematical Sciences, Nanyang Technological University, Singapore}
\affiliation{Science Institute, University of Iceland, Dunhaga 3, IS-107 Reykjavik, Iceland}

\author{Mauro Paternostro}
\affiliation{School of Mathematics and Physics, Queen's University, Belfast BT7 1NN, United Kingdom}

\author{Tomasz Paterek}
\affiliation{School of Physical and Mathematical Sciences, Nanyang Technological University, Singapore}
\affiliation{Centre for Quantum Technologies, National University of Singapore, Singapore}
\affiliation{MajuLab, CNRS-UNS-NUS-NTU International Joint Research Unit, UMI 3654, Singapore}


\begin{abstract}
Some physical objects are hardly accessible to direct experimentation.
It is then desirable to infer their properties based solely on the interactions they have with systems over which we have control.
In this spirit, here we introduce schemes for assessing the non-classicality of the inaccessible objects as characterised by quantum discord.
We consider two probes individually interacting with the inaccessible object, but not with each other.
The schemes are based on monitoring entanglement dynamics between the probes.
Our method is robust and experimentally friendly as it allows the probes and the object to be open systems, makes no assumptions about the initial state, dimensionality of involved Hilbert spaces and details of the probe-object Hamiltonian.
We apply our scheme to a membrane-in-the-middle optomechanical system, to detect system-environment correlations in open system dynamics as well as non-classicality of the environment,
and we foresee potential benefits for the inference of the non-classical nature of gravity.
\end{abstract}

\maketitle

What should be known about an inaccessible object to conclude that it is ``not classical"? 
Here we show, inspired by quantum communication scenarios, that it is sufficient to verify whether such object can be used to increase quantum entanglement between remote probe-particles that individually interact with it, but are not directly coupled to each other.

Specifically, we prove that such gain in quantum entanglement is only possible if, during its evolution, the object shares with the probes quantum correlations in the form of quantum discord~\cite{henderson2001classical,discord,celeri-review,reviewd,JPA.49.473001}.
In turn, the presence of quantum discord between the probes and the object entails a non-classical feature of the object itself.
According to the definition of discord, two or more subsystems share quantum correlations if there is no von Neumann measurement on one of them that keeps the total state unchanged.
This can only happen when non-orthogonal (indistinguishable) states are involved in the description of the physical configuration of the measured subsystem.
This indistinguishability is the non-classical feature that we aim to detect. We formulate analytical criteria revealing such non-classicality based on operations performed only on the probes, and without any detailed modelling of the inaccessible object in question.

We emphasise that the non-classicality is revealed under a set of minimal assumptions.
Namely: (i) The object may remain inaccessible at all times, i.e. it needs not be directly measured. 
In particular its quantum state and Hilbert space dimension can remain unknown throughout the whole assessment. 
Our method is thus valid when the object is an elementary system or an arbitrarily complex one;
(ii) The details of the interaction between the object and the probes may also remain unspecified;
(iii) Every party can be open to its own local environment.
These properties make our method applicable to a large number of experimentally relevant situations. 

We demonstrate the revealing power of our criteria for non-classicality through the study of an optomechanical system, which is a platform of enormous experimental interest.
This is clearly not the only situation that can benefit from the results of our investigation.
We conclude the paper with a discussion of a set of physical problems, from the revelation of system-environment correlations in open system dynamics to the quest for the possible quantum nature of gravity, that would be fully suited to the framework presented here.

\begin{figure}[!b]
\includegraphics[width=0.4\textwidth]{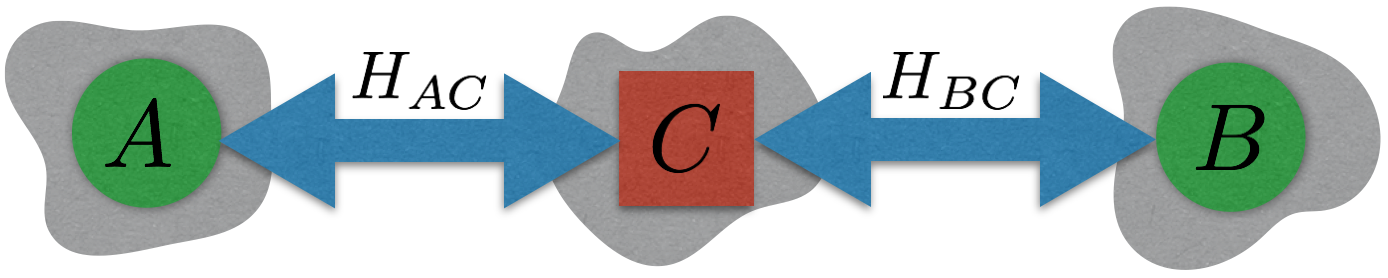} 
\caption{Probes $A$ and $B$ individually interact with a mediator object $C$, but not with each other. 
We allow $C$ to be inaccessible, i.e. no measurement can be performed on it and its state may remain unknown. 
We show conditions under which the gain of entanglement in $AB$ implies non-zero quantum discord $D_{AB|C}$. 
Our protocols make no assumptions about the dimensions of each subsystem, explicit form of $H_{AC}$ and $H_{BC}$, and allow each subsystem to be open to its own environment (represented by grey-colored shadows).}
\label{FIG_ABC} 
\end{figure}

\emph{The formal criteria.} Consider the scenario depicted in Fig.~\ref{FIG_ABC}.
System $C$ is assumed to be the inaccessible object and to mediate the interaction between two remote probes, labeled $A$ and $B$. Therefore, from now on, we refer to system $C$ as the {\it mediator}.
It is essential for our method that the probes are not directly coupled and only interact via the mediator. Therefore, the Hamiltonian for the process under scrutiny can be written as $H_{AC} + H_{BC}$, with $H_{JC}$ the interaction Hamiltonian between the mediator $C$ and probe $J=A,B$.
Our work is developed in the context of entanglement distribution with continuous interactions~\cite{cubitt}.
We first focus on the partition $A:BC$ and demonstrate a result which will be instrumental to design our criteria for the inference of non-classicality of $C$ based on entanglement dynamics in $AB$ only.
Previous studies on the resources allowing for entanglement distribution showed that any three-body density matrix, i.e. the state of $ABC$ at any time $t$ in the present context, satisfies the inequality~\cite{bounds1, bounds2}:
\begin{equation}
\label{ebound}
|E_{A:BC}(t) - E_{AC:B}(t) | \le D_{AB|C}(t).
\end{equation}
Here $E_{X:Y}$ is the relative entropy of entanglement in the partition $X:Y$~\cite{vedral1997quantifying}, and $D_{X|Y}$ is the relative entropy of discord~\cite{modi2010unified}, also known as the one-way quantum deficit~\cite{deficit}.
Note that relative entropy of discord is in general not symmetric, i.e. $D_{X|Y} \ne D_{Y|X}$.
Eq.~(\ref{ebound}) shows that the change in entanglement due to the relocation of $C$ is bounded by the quantum discord carried by it.

Let us start from the simple case where the overall probes-mediator system is closed (which allows us to ignore for now the grey-colored shadows in Fig.~\ref{FIG_ABC}).
If the interaction Hamiltonians $H_{JC}$ satisfy $[H_{AC},H_{BC}] = 0$, the evolution operator from the initial time $t=0$ to some finite time $\tau$ is just 
$U = U_{BC} U_{AC}$, where $U_{JC} = \exp{(- i H_{JC} \tau)}$ and we set $\hbar = 1$.
This situation is equivalent to first interacting $C$ with $A$ and then $C$ with $B$ (or in reversed order).
However, note that the density matrix $\rho' = U_{AC} \rho_0 U_{AC}^\dagger$ obtained by ``evolving'' the initial state through $U_{AC}$ only does not describe the state of the system at $\tau$.
Nevertheless, we now show the relevance of the properties of state $\rho'$ for entanglement gain.

Consider the following forms of Eq.~(\ref{ebound}) written for the initial state $\rho_0$ and the instrumental state $\rho'$, respectively
\begin{equation}
\begin{aligned}
&E_{AC:B}(0) - E_{A:BC}(0) \le D_{AB|C}(0), \\
&E_{A:BC}' - E_{AC:B}' \le D_{AB|C}'.
\end{aligned}
\end{equation}
Note that $E_{AC:B}(0) = E_{AC:B}'$, because $U_{AC}$ is local in this partition.
The state at time $\tau$ is given by $\rho_\tau = U_{BC} \rho' U_{BC}^\dagger$, and thus $E_{A:BC}(\tau) = E_{A:BC}'$, this time owing to $U_{BC}$ being local.
Summing the above inequalities we obtain a bound on the entanglement gain
\begin{equation}
\label{ebound_finite}
E_{A:BC}(\tau) - E_{A:BC}(0) \le D_{AB|C}(0) + D_{AB|C}'.
\end{equation}
This opens up the possibility to create entanglement at time $\tau$ without producing discord at both $t=0$ and $\tau$, but rather by utilising non-classicality in the instrumental state.
In other words, the gain of entanglement in $A:BC$ could be mediated by object $C$, which gets non-classically correlated by $U_{AC}$ 
and then decorrelated by $U_{BC}$. Therefore, $C$ is only classically correlated at times $t=0$ and $\tau$. 
We now give a concrete example of this type of entanglement creation.

Consider the interaction Hamiltonian
\begin{equation}
\label{xhamiltonian}
H=\sigma^x_A\otimes \openone \otimes \sigma^x_C+\openone \otimes \sigma^x_B \otimes \sigma^x_C ,
\end{equation}
where $\sigma^j~(j=x,y,z)$ is the Pauli-$j$ matrix. 
As initial state we choose the classically correlated state
\begin{eqnarray}
\rho_0 &=& \tfrac{1}{2}| 011 \rangle \langle 011 | + \tfrac{1}{2}| 100 \rangle \langle 100 |,
\end{eqnarray}
where e.g. $\sigma^z|0\rangle=|0\rangle$. 
One can now readily check that the relative entropy of entanglement $E_{A:BC}$ grows from $0$ to $1$ in the timespan from $t=0$ to $\tau = \pi/4$, whereas discord $D_{AB|C}$ remains zero at these two times. The gain is indeed due to non-classical correlations of the instrumental state: applying only $U_{AC}$ for a time $\tau = \pi/4$ produces discord $D_{AB|C}' = 1$.

For general non-commuting interaction Hamiltonians, one can pursue a similar analysis with the help of the Suzuki-Trotter expansion.
The evolution operator $U$ is now discretised into short-time interactions of $C$ with $A$ and then $B$ (or the reversed order) as
\begin{eqnarray}
U&=& \lim_{n\rightarrow \infty} \left(e^{-i H_{BC}\Delta{t}}e^{-i H_{AC}\Delta{t}}\right)^n ,
\end{eqnarray}
where $\Delta t = \tau /n \rightarrow 0$. 
Accordingly, Eq.~(\ref{ebound_finite}) holds with $\tau$ replaced by $\Delta t$.
It is now natural to ask if a scenario exists where entanglement could be increased via interactions with a classical $C$ \emph{at all times} by exploiting the discord in the instrumental state.
The example given above is not of this sort because, although we have $D_{AB|C}=0$ at $t=0$ and $\tau$, it is non-zero for $t \in (0,\tau)$.
It turns out that, for short evolution times, the discord of the instrumental state cannot be exploited as the following Theorem demonstrates.

\begin{theorem*}
For three open systems $A$, $B$, $C$ with Hamiltonian $H = H_{AC} + H_{BC}$ and each coupled to its own local environment, the entanglement satisfies the condition $E_{A:BC}(\tau)\le E_{A:BC}(0) $ if $D_{AB|C}(t) = 0$ at any time $t\in [0,\tau]$.
\end{theorem*}
\begin{proof}
The proof is presented in Appendix~\ref{SEC_PROOF}. 
\end{proof}

We emphasise the generality of this Theorem, where both mediator and probes are open to their own local environments. This matches a large number of experimentally relevant situations, some of them being addressed in the last part of this paper.
The setup where $A$, $B$, and $C$ are closed systems is then a special case of the Theorem above in which we have $E_{A:BC}(\tau)=E_{A:BC}(0)$ if $D_{AB|C}(t)=0$ (see Appendix~\ref{SEC_PROOF}). 
Such Theorem extends the monotonicity of entanglement under local operations and classical communication (LOCC)~\cite{locc} to the case of continuous interactions.
In general, zero-discord states are good models for classical communication as they allow for continuous projective measurements on $C$ that do not disturb the whole multipartite state.

We are now in a position to study the presence of discord $D_{AB|C}$ from observing $AB$ only.
In light of the Theorem above, a promising candidate for this goal is the entanglement gain.
However, we now show that some features of the initial tripartite state need to be ensured, but they can be guaranteed by only operating on $AB$.

Let us consider Eq.~(\ref{xhamiltonian}) and choose the initial state
\begin{equation}
\rho_0 = \tfrac{1}{2} \ket{\psi_+}\bra{\psi_+} \otimes \ket{+}\bra{+} + \tfrac{1}{2} \ket{\phi_+}\bra{\phi_+} \otimes \ket{-}\bra{-},
\label{EQ_CL_EX}
\end{equation}
where $\sigma^x\ket{\pm}=\pm|\pm\rangle$, and $\ket{\psi_+}= \frac{1}{\sqrt{2}}(\ket{01}+\ket{10})$ and $\ket{\phi_+} = \frac{1}{\sqrt{2}}(\ket{00}+\ket{11})$ are two Bell states between subsystems $AB$.
As the initial state in Eq.~(\ref{EQ_CL_EX}) contains the eigenstates of $H_C$, the system remains classical, as measured on $C$, at all times.
Furthermore, the classical basis is the same at all times.
Yet, one can verify that the relative entropy of entanglement between the probes is given by $E_{A:B}(t) = 1-S_{AB}(t)$, where $S_{AB}(t)$ is the von Neumann entropy of the $AB$ state at time $t$, and oscillates between $0$ and $1$.
Hence, in general, entanglement gain in the partition $A:B$ does \emph{not} signify the non-classicality of $C$ (non-zero $D_{AB|C}$).

Similar considerations have been presented in Ref.~\cite{gyongyosi2014correlation} to provide a counter-example of the impossibility of entanglement gain via LOCC.
However, the partition $A:BC$ is entangled already from the beginning (in our example we have $E_{A:BC} = 1$).
The subsequent evolution only localises such entanglement to the $A:B$ partition.
This example emphasises that the ancillary particles within the framework of LOCC, here $C$, are not allowed to be initially correlated with the principal system, here $AB$, even if the correlations are classical.

Furthermore, the only way of gaining entanglement in subsystem $AB$ via classical $C$ is to localise it from the already present entanglement in $A:BC$.
This is a consequence of our Theorem and reinforces its role as a proper generalisation of the monotonicity of entanglement to continuous interactions.
Namely,
\begin{equation}
E_{A:B}(\tau) \le E_{A:BC}(\tau) \le E_{A:BC}(0).
\label{E_CHAIN}
\end{equation}
Now, if we ensure by operating on the probes only that the initial entanglements coincide, i.e. $E_{A:BC}(0) = E_{A:B}(0)$, entanglement gain in system $AB$ is only possible due to non-zero discord $D_{AB|C}$.
As we are interested in observing entanglement gain, it is natural to start with as small entanglement as possible.
This leads us to propose the application of an entanglement-breaking channel to one of the available systems, at time $t=0$.
Indeed, after application of the channel, we have $E_{A:B}(0)=E_{A:BC}(0)=0$.
In a more concrete example, the channel is a von Neumann measurement.
An arbitrary measurement is allowed and experimentalist should choose the one having potential for biggest entanglement gain.
Note that the measurement results need not be known.
Our main detection method is illustrated and summarised in Fig.~\ref{FIG_p2}.
We stress that entanglement estimation in step (iii) can be realised with entanglement witnesses~\cite{horo_witness,RevModPhys.81.865}, rendering state tomography unnecessary.
(See Appendix~\ref{SEC_CRIT} for a criterion based on comparison between entanglement and initial purities of the probes.)

\begin{figure}[!t]
\includegraphics[width=0.3\textwidth]{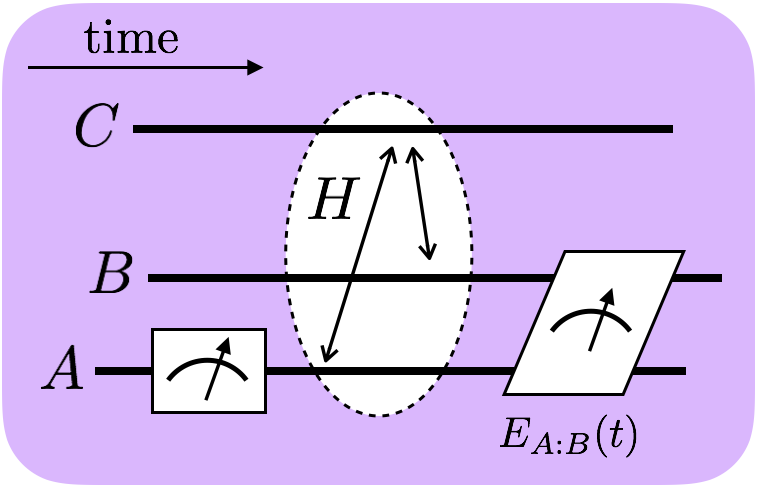} 
\caption{Proposed protocol detecting non-classicality of inaccessible object $C$.
Each system $A$, $B$, and $C$ can be open system, with its own local environment. 
The protocol makes no assumptions about the initial tripartite state and the explicit expression for the Hamiltonian $H_{AC}+H_{BC}$,
and has the following steps:
(i) von Neumann measurement in subsystem $A$ (or any entanglement-breaking channel);
(ii) evolution of the whole $ABC$;
(iii) entanglement estimation of $AB$.
We show in the main text that nonzero entanglement reveals positive discord $D_{AB|C}$.}
\label{FIG_p2} 
\end{figure}

{\it Optomechanics.} We address now the practical implications of our criteria for scenarios of current technological relevance. In particular, we consider experiments of cavity optomechanics~\cite{RMP.86.1391} as the paradigm of an open mesoscopic quantum system for which the criteria identified above hold the potential to be practically significant. In fact, one of the goals of optomechanics is to infer the non-classicality of the state of a massive mechanical system without affecting its (in general fragile) state. A possible setting for such a task is given by a so-called membrane-in-the-middle configuration, where a mechanical oscillator (\emph{a membrane}) is suspended at the centre of a two-sided optical cavity~\cite{paternostro}.
By driving the cavity with laser fields from both its input mirrors, respectively, we realise a situation completely analogous to that in Fig.~\ref{FIG_ABC} (cf. Fig.~\ref{FIG_p3}).
We now show that our scheme detects non-classicality of the membrane without measuring it.

\begin{figure}[!b]
\includegraphics[width=0.34\textwidth]{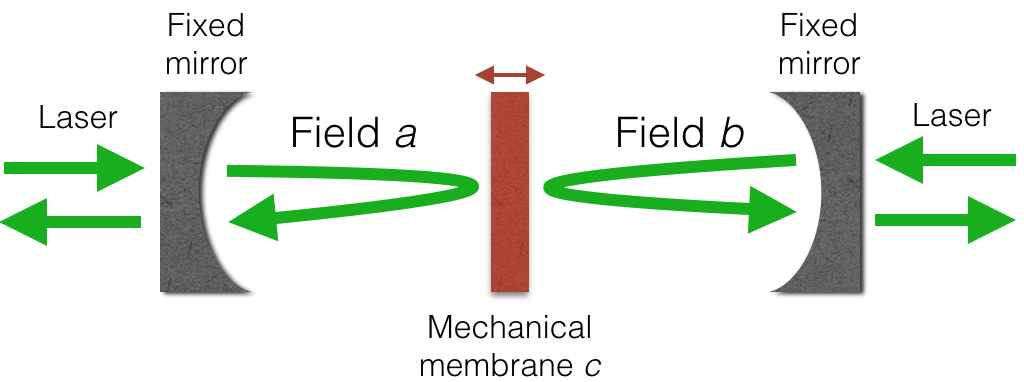} 
\caption{Optomechanics setup. The mechanical membrane $c$ is mediating interaction between driven cavity fields $a$ and $b$.
The membrane is interacting with its local environment at temperature $T$ resulting in the Brownian motion
and the fields are independently interacting with their respective driving lasers through the fixed mirrors.
}
\label{FIG_p3} 
\end{figure}

The interaction Hamiltonian for the setup in Fig.~\ref{FIG_p3} reads~\cite{paternostro}
\begin{equation}
H_{\mathrm{int}} = - \hbar G_{0a} \, a^\dagger a \, q + \hbar G_{0b} \, b^\dagger b \, q.
\end{equation}
This is complemented by local terms affecting each subsystem individually (cf. Appendix~\ref{SEC_OPTMECH}).
Here $a$ and $b$ are the annihilation operators for the respective fields, 
$q$ is the dimensionless position-like quadrature of the membrane,
and $G_{0a(b)}$ represents the strength of the coupling between field $a$ ($b$) and the membrane.
All the other interactions are local, i.e. $a$ ($b$) is coupled to its own environment $a'$ ($b'$) and $c$ is coupled to its thermal phonon reservoir $c'$, responsible for the Brownian motion of the membrane.
Thus, our Theorem directly applies here and we can implement the detection method of Fig.~\ref{FIG_p2}.

In order to independently confirm the non-classicality of the membrane and demonstrate that there is considerable entanglement to be detected we now calculate the ensuing entanglement dynamics.
We choose the logarithmic negativity to quantify entanglement. 
Starting from the experimentally natural state where $c$ is in a thermal state and $a$ and $b$ are coherent states, we calculate the dynamics of $E_{a:b}$ and $E_{ab:c}$.
As initially there is no entanglement, the first step in Fig.~\ref{FIG_p2} can be omitted.
The results of our analysis are presented in Fig.~\ref{FIG_OPTORES} for varying power of the right laser. The parameters used in our simulations all adhere to present-day technology \cite{groblacher2009observation}.
We see that non-zero $E_{a:b}(\tau)$ is always accompanied by non-zero $E_{ab:c}$ at some time $(0,\tau)$.
Note that entanglement is a stronger type of quantum correlations than discord.
We have also performed similar calculations by varying the power of the left laser as well as the frequencies of the lasers within experimentally accessible ranges and observed consistent results (see Appendix~\ref{SEC_OPTMECH}).

\begin{figure}[t]
\includegraphics[scale=0.3]{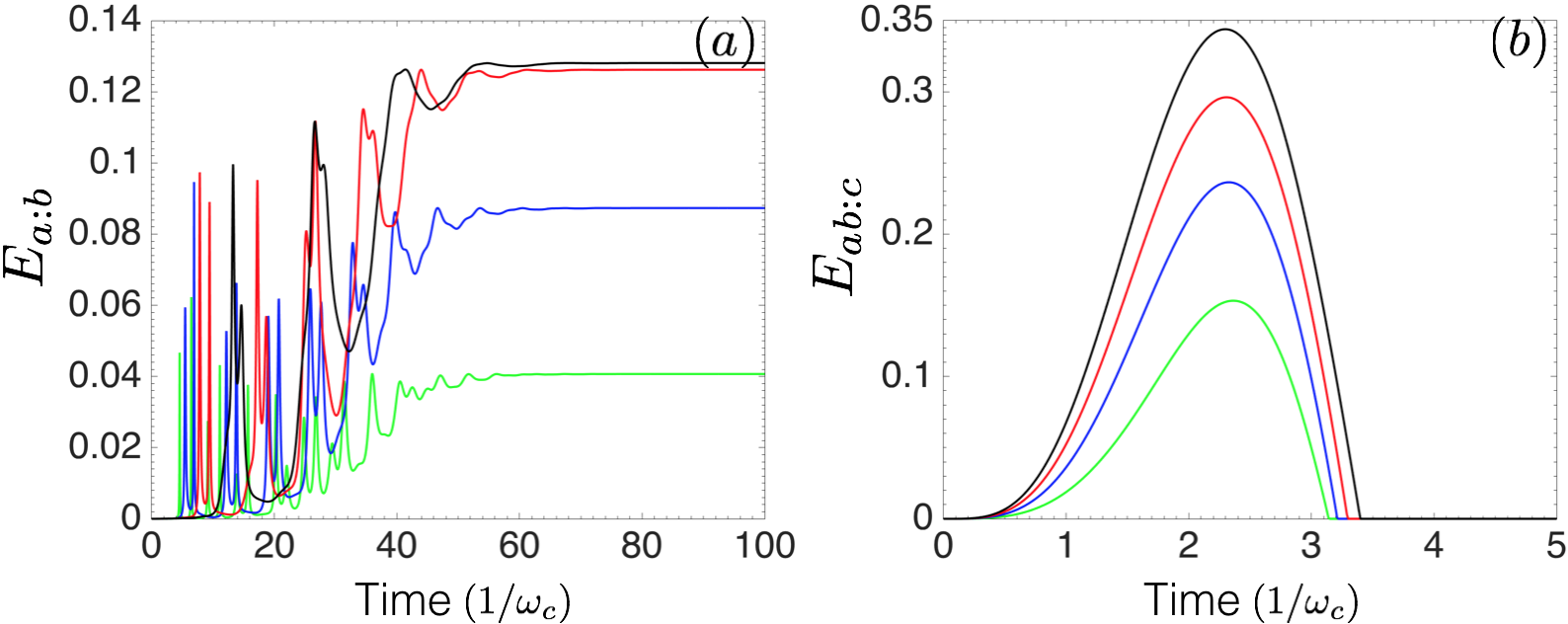}
\caption{Exemplary dynamics of entanglement (logarithmic negativity) $E_{a:b}$ and $E_{ab:c}$ for experimentally viable parameters. 
Mass of membrane $145$ ng with damping rate $2\pi \times 140$ Hz, temperature $0.3$ K, length of each cavity $25$ mm with finesse $1.4\times10^4$, and wavelength of both lasers is $1064$ nm. Here we fixed $P_a=100$ mW, $\Delta_a=\omega_c$, and $\Delta_b=-\omega_c$, where $\omega_c=2\pi \times947$ kHz is the natural frequency of the membrane. We vary $P_b=20$ mW (green), $40$ mW (blue), $60$ mW (red), and $80$ mW (black). $P_{a(b)}$ stands for the power of the left (right) laser, and $\Delta_{a(b)}$ is its effective detuning (see Appendix~\ref{SEC_OPTMECH} for detailed calculations). Note that non-zero entanglement between the fields implies that the membrane is entangled with them in the process.}
\label{FIG_OPTORES} 
\end{figure}

\emph{System-environment correlations.}
As a second relevant application of our study, let us consider again a closed-system dynamics and, in line with the assumed inaccessibility of the mediator, focus the attention to the probes only.
We could thus think of $C$ as an environment in contact with the open system $AB$.
A vast body of literature exists on the study of the influence of initial system-environment correlations (SECs) on the evolution of the open system~\cite{breuer2002theory}.
Proposals for the detection of SECs based on monitoring the dynamics of distinguishability \cite{laine2011witness, smirne2010initial, dajka2010distance, dajka2011distance, wissmann2013detecting} or purity \cite{kimura2007,rossatto2011purity} of the open system have been put forward.
Such proposals have been implemented experimentally by means of quantum tomography \cite{smirne2011experimental, li2011experimentally}.
Moreover, the possible non-classical nature of SECs was linked to the impossibility of describing the evolution of an open system through completely positive maps~\cite{rodriguez2008completely}.
Hence detection schemes of quantum discord in the initial system-environment state have been proposed~\cite{gessner2011detecting, gessner2013local} 
and recently assessed experimentally \cite{gessner2014local, tang2015experimental, cialdi2014two}. 

Our scheme of Fig.~\ref{FIG_p2} can also be used to reveal SECs, with the advantage that state tomography is not necessary. 
This is achieved by dividing the open system into $A$ and $B$ parts and monitoring the presence of entanglement between them.
If one is only interested in the detection of correlations between $AB$ and $C$, regardless of whether they are classical or not, the entanglement-breaking channel in Fig.~\ref{FIG_p2} can be omitted.
Indeed, for the initially uncorrelated state $\rho_0 = \rho_{AB} \otimes \rho_C$, we have $E_{A:BC}(\rho_{AB} \otimes \rho_C) = E_{A:B}(\rho_{AB})$, and no entanglement gain in $AB$ is possible via classical $C$.
Therefore, if one observes a gain, it would either be $\rho_0 \ne \rho_{AB} \otimes \rho_C$ or $D_{AB|C} > 0$ at some time. Both cases show correlations between $AB$ and $C$. 
Finally, we note that previous schemes detect the non-classicality of the system \cite{gessner2011detecting, gessner2013local}, i.e. presence of $D_{C|AB}$, whereas our schemes ascertain the non-classicality of the environment, $D_{AB|C}$,
which is perhaps a prime example of an inaccessible object.

\emph{Other applications.} A similar analysis can be done for remote quantum dots in a solid-state substrate~\cite{pp} or spin-chain systems like in Ref.~\cite{NatPhys.11.255},
as their physics also naturally distinguishes a mediating object that is inaccessible, e.g. locations of unpaired spins are unknown in a sample~\cite{NatPhys.11.255}.
In a visionary perspective, system $C$ could even be a gravitational field coupling massive systems $A$ and $B$, which are mutually non-interacting.
By determining experimentally the entanglement gain between $A$ and $B$ one would conclude, according to our scheme, the non-classical nature of the gravitational field between them.
That is, if we were to embed into the quantum formalism description of the masses and the field, there would have to be non-orthogonal states in the Hilbert space of the field
as this is required for the quantum discord $D_{AB|C}$ to be non-zero.

{\it Conclusions.} We have proposed an entanglement-based criteria for the inference of non-classicality of an inaccessible object.
Our protocols are fully non-disruptive of the state of the system to probe, and rely on only weak assumptions on the nature of the interactions involved.
They are also robust against decoherence.
These features make our proposal suitable to address non-classicality at many levels, 
from experimentally relevant technological platforms such as quantum optomechanics, to fundamental problems on the nature of gravity. 

{\it Acknowledgments.}
We thank {\v C}. Brukner, K. Modi, M. Piani and A. Winter for stimulating discussions.
This work is supported by the National Research Foundation and Singapore Ministry of Education Academic Research Fund Tier 2 project MOE2015-T2-2-034.
MP acknowledges support from the EU project TherMiQ, the John Templeton Foundation (grant nr. 43467), the Julian Schwinger Foundation (grant nr. JSF-14-7-0000).

\appendix

\section{Proof of the theorem}
\label{SEC_PROOF}

\begin{theorem*}
For three open systems $A$, $B$, $C$ with hamiltonian $H = H_{AC} + H_{BC}$ and each coupled to its own local environment, the entanglement satisfies the condition $E_{A:BC}(\tau)\le E_{A:BC}(0) $ if $D_{AB|C}(t) = 0$ at any time $t\in [0,\tau]$.
\end{theorem*}
\begin{proof}
Without loss of generality we take the hamiltonians as $H_{AC} = H_A \otimes H_C$ and $H_{BC} = H_B \otimes H_\gamma$.
The same steps can be applied to the general hamiltonians $H_{AC} = \sum_{\mu} H_A^{\mu} \otimes H_C^{\mu}$ and $H_{BC} = \sum_{\nu} H_B^{\nu} \otimes H_\gamma^{\nu}$.
As we are assuming vanishing discord (when $C$ is measured) at $t=0$, we can write the initial state as
\begin{equation}
\rho_0 = \sum_{c} p_c \: \rho_{AB}^c \otimes \ket{c}\bra{c}.
\end{equation}
After a time $\Delta t$, the evolved state will be isomorphic to $\rho_0$ and will read
\begin{equation}
\rho_{\Delta t} = \sum_{c} p_c(\Delta t) \: \rho_{AB}^c(\Delta t) \otimes \ket{\phi_c}\bra{\phi_c},
\end{equation}
where both $\{\ket{c}\}$ and $\{ \ket{\phi_c} \}$ form orthonormal bases.
We assume local environments inducing Markovian dynamics. The evolution can thus be described by the following master equation in Lindblad form (which we write, for convenience, in a coarse-grained time form)
\begin{equation}\label{EQ_LE}
\frac{\rho_{\Delta t}-\rho_0}{\Delta t}=-i[H,\rho_0]+\sum_{X=A,B,C}L_X \rho_0,
\end{equation}
where the last term is the incoherent part of the evolution resulting from interactions with local environments, and we have taken (in light of the declared Lindblad form of the master equation) $L_X\rho_0=\sum_k Q^X_k\rho_0 Q^{X\dag}_k-\frac{1}{2}\{Q^{X\dag}_kQ^X_k,\rho_0\}$,
where $Q^X_k$'s act on system $X$ only. We now consider the following conditional state
\begin{eqnarray}
p_j(\Delta t) \: \rho_{AB}^j(\Delta t) &=&\bra{\phi_j}  \rho_{\Delta t}\ket{\phi_j}\nonumber \\
&=& \bra{\phi_j} \rho_0 \ket{\phi_j} -i \Delta t \bra{\phi_j} [H,\rho_0] \ket{\phi_j}\nonumber \\
&&+\Delta t  \sum_{X=A,B,C} \bra{\phi_j} L_X \rho_0\ket{\phi_j}.
\label{EQ_COMM}
\end{eqnarray}
We go to continuous time by taking $\Delta t \to 0$, which implies that all terms proportional to ${\cal O}(\Delta t^2)$ can be ignored.
Therefore, throughout this proof the symbol ``$\simeq$'' should be read as ``equal up to ${\cal O}(\Delta t^2)$''.
In this limit, the basis states of $C$ can only change as $\ket{c} \to \ket{\phi_c} = \alpha_c \ket{c} + \beta_c \ket{c_\perp}$, with $\ket{c_\perp}$ orthogonal to $\ket{c}$ and $\beta_c$ proportional to $\Delta t$.
This implies that $|\langle c | \phi_j \rangle|^2 \simeq \delta_{cj}$, and thus $\bra{\phi_j} \rho_0 \ket{\phi_j} \simeq p_j \rho_{AB}^j$.

As the commutator in Eq.~(\ref{EQ_COMM}) is multiplied by $\Delta t$, the terms in $\bra{\phi_j} [H,\rho_0] \ket{\phi_j}$ proportional to $\Delta t$ can already be ignored. One thus finds that
\begin{equation}
 \bra{\phi_j} [H,\rho_0] \ket{\phi_j} \simeq p_j [E_C^j H_A + E_\gamma^j H_B,\rho_{AB}^j] ,
\end{equation}
where $E_{C(\gamma)}^j$ is the mean energy $\bra{j} H_{C(\gamma)} \ket{j}$.
Hence the coherent part of the evolution splits into the sum of effective local interactions.

Next, there are three terms for the incoherent part of the evolution.
The first two are local in $A$ and $B$ respectively, and they are proportional to
\begin{equation}
\bra{\phi_j} L_A \rho_0+L_B \rho_0\ket{\phi_j}\simeq p_j(L_A+L_B)\rho_{AB}^j.
\end{equation}
The last term can be written as 
\begin{eqnarray}
\bra{\phi_j} L_C \rho_0\ket{\phi_j}&\simeq&\sum_c \sum_k p_c |\bra{j}Q^C_k\ket{c}|^2 \rho_{AB}^c\nonumber \\
&-&p_j\sum_k  \bra{j}Q_k^{C\dagger}Q^C_k\ket{j} \rho_{AB}^j.
\end{eqnarray}

Therefore, we obtain the following explicit form of the conditional state
\begin{eqnarray}
\rho_{AB}^j(\Delta t)& \simeq &  \frac{p_j(1-\sum_k  \bra{j}Q_k^{C\dagger}Q^C_k\ket{j}\Delta t )}{p_j(\Delta t)} \, \, \tilde \rho_{AB}^j \nonumber \\
  & + & \sum_c\frac{p_c}{p_j(\Delta t)}\sum_k |\bra{j}Q^C_k\ket{c}|^2\Delta t \: \rho_{AB}^c,
\label{EQ_US}
\end{eqnarray}
where we have defined 
\begin{eqnarray}
\tilde \rho_{AB}^j & \equiv & \rho_{AB}^j -i[E_C^j H_A + E_\gamma^j H_B,\rho_{AB}^j] \Delta t \nonumber \\
& + & (L_A+L_B)\rho_{AB}^j \Delta t.
\end{eqnarray}
Accordingly, the state $\tilde \rho_{AB}^j$ is obtained from $\rho_{AB}^j$ by applying Lindblad master equation describing independent local evolutions of subsystems $A$ and $B$.
The probability $p_j(\Delta t)$ can be recovered by noting that the trace of Eq.~(\ref{EQ_US}) equals unity. In particular,
\begin{eqnarray}
p_j(\Delta t)&=& p_j \left( 1-\sum_k  \bra{j}Q_k^{C\dagger}Q^C_k\ket{j}\Delta t \right)\nonumber \\
  &&+\sum_cp_c\sum_k |\bra{j}Q^C_k\ket{c}|^2\Delta t,
\end{eqnarray}
where we have used the cyclic property of trace.
Note that the numbers multiplying states $\tilde \rho_{AB}^j$ and $\rho_{AB}^c$ in Eq. (\ref{EQ_US}) are all non-negative and sum up to unity, i.e. they form a probability distribution.

Finally, we have
\begin{eqnarray}
E_{A:BC}(\Delta t)&=&\sum_{j} p_{j}(\Delta t) \: E_{A:B}(\rho_{AB}^{j}(\Delta t)) \\
&\le&\sum_j p_j E_{A:B}(\rho_{AB}^j) \label{EQ_LL} \\
&-&\sum_jp_j\sum_k \bra{j}Q_k^{C\dagger}Q^C_k\ket{j}\Delta t \:E_{A:B}(\rho_{AB}^j)\nonumber \\
&+&\sum_j\sum_cp_c\sum_k|\bra{j}Q^C_k\ket{c}|^2\Delta t \:E_{A:B}(\rho_{AB}^c),\nonumber \\
&=& \sum_j p_j E_{A:B}(\rho_{AB}^j) = E_{A:BC}(0), \label{EQ_LAST}
\end{eqnarray}
where the involved steps are justified as follows.
In the first line we use the flags condition~\cite{flags}.
The inequality follows from applying convexity of the relative entropy of entanglement to $E_{A:B}(\rho_{AB}^{j}(\Delta t))$ 
and next monotonicity of entanglement under local operations and classical communication $E_{A:B}(\tilde \rho_{AB}^j)\le E_{A:B}(\rho_{AB}^j)$. 
By inserting $\sum_c \ket{c}\bra{c}=\openone$ in between $Q_k^{C\dagger}$ and $Q_k^C$ in the second line of (\ref{EQ_LL}) and exchanging dummy indices $c \leftrightarrow j$,
one finds that it cancels the third line, giving Eq.~(\ref{EQ_LAST}).
In the last step we again use the flags condition, this time to the initial state.
The theorem is apparent by evolving the system successively from $t=0$ to $\tau$ and having $D_{AB|C}(t)=0$ at anytime $t\in [0,\tau]$.
\end{proof} 

A special case worth noticing is when the three systems are closed, i.e. we have unitary evolution with hamiltonian $H=H_{AC}+H_{BC}$. 
In this case, it follows that the conditional state reads
\begin{equation}\label{EQ_SC}
\rho_{AB}^j(\Delta t)=\frac{p_j}{p_j(\Delta t)} \, \, \tilde \rho_{AB}^j,
\end{equation}
where $\tilde \rho_{AB}^j = \rho_{AB}^j-i[E_C^j H_A + E_\gamma^j H_B,\rho_{AB}^j] \Delta t$.
Taking the trace of Eq.~(\ref{EQ_SC}) gives us $p_j(\Delta t)=p_j$ and $\rho_{AB}^j(\Delta t)=\tilde \rho_{AB}^j$.
The conditional state evolves under effective local unitary transformations, hence entanglement in the partition $A:B$ stays the same. 
By utilising the flags condition as in the proof of the theorem given above, one can show that $E_{A:BC}(\tau) = E_{A:BC}(0)$.


\section{A simple criterion based on purity}
\label{SEC_CRIT}

An elegant criterion for revealing the non-classicality can be derived in terms of the initial purity of probes.
Here we take the purity to be given by the von Neumann entropy.
Our starting point is Eq. (8) of the main text, which is here repeated for convenience:
\begin{equation}
E_{A:B}(\tau) \le E_{A:BC}(\tau) \le E_{A:BC}(0).
\end{equation}
We will now bound the initial entanglement $E_{A:BC}(0)$ by the purity of the probes only.
The relative entropy of entanglement is upper bounded by the mutual information $E_{A:BC} \le I_{A:BC}$~\cite{modi2010unified}. 
From the sub-additivity of entropy for the $BC$ subsystem we have $I_{A:BC} \le S_A + S_B - S_{AB|C}$, and as the state of $C$ is assumed to be classical, $S_{AB|C} \ge 0$, we get $E_{A:BC}(0) \le S_A(0) + S_B(0)$.
Accordingly,
\begin{equation}
E_{A:B}(\tau) \le S_A(0) + S_B(0).
\end{equation}
Any violation of this inequality reveals the non-classicality of $C$,
i.e. non-zero $D_{AB|C}$ at some time during the evolution.


\section{Detailed assessment of the optomechanical application}
\label{SEC_OPTMECH}

The hamiltonian of the setup (Fig. 3 in the main text) in a rotating frame with frequency of the lasers can be written as $H=H_{\mbox{loc}}+H_{\mbox{int}}$ where \cite{RMP.86.1391}:
\begin{eqnarray}\label{Eq_hloc}
H_{\mbox{loc}}&=&\hbar \Delta_{0a} a^{\dagger} a+\hbar \Delta_{0b} b^{\dagger} b+\frac{\hbar \omega_c}{2}( p^2+ q^2) \nonumber \\ 
&&+i\hbar E_a(a^{\dagger}- a)+i\hbar E_b( b^{\dagger}- b) 
\end{eqnarray}
and 
\begin{equation}\label{Eq_hint}
H_{\mbox{int}}=-\hbar G_{0a} a^{\dagger}a \:q+\hbar G_{0b}  b^{\dagger} b \:q,
\end{equation}
where the annihilation (creation) operator of field $j=a,b$ is denoted by $ j$ ($ j^{\dagger}$) with $[ j, j^{\dagger}]=1$, 
$p$ and $ q$ are dimensionless quadratures of the membrane with $[q, p]=i$, 
$E_j$ is the driving strength of laser $j$ with $|E_j|=(2P_j\kappa_j/\hbar \omega_{lj})^{1/2}$, where $P_j$ is the laser power and $\omega_{lj}$ denotes its frequency. $\kappa_j=\pi c/2F_jl_j$ is decay rate of cavity $j$ with finesse $F_j$.
Cavity-laser detuning is defined as $\Delta_{0j}\equiv \omega_j-\omega_{lj}$, where $\omega_j$ is the frequency of the cavity
and $G_{0j}= (\omega_j/l_j)(\hbar/\mu \omega_c)^{1/2}$ represents field-membrane coupling strength, 
where $l_j$ is the length of the cavity, $\mu$ is the mass of the membrane and $\omega_c$ is its natural frequency.
Note that $H_{\mbox{loc}}$ is local in $a:b:c$ partition and the two terms in $H_{\mbox{int}}$ represent coupling in the partition $a:c$ and $b:c$ respectively. 

The dynamics of the operators, adding into account noise and damping terms (also local), can be well written by a set of Langevin equations in Heisenberg picture
\begin{eqnarray}\label{Eq_lgvn}
\dot a&=&-(\kappa_a+i\Delta_{0a})a+iG_{0a}aq+E_a+\sqrt{2\kappa_a}\:a_{\mbox{in}} \nonumber \\
\dot b&=&-(\kappa_b+i\Delta_{0b})b-iG_{0b}bq+E_b+\sqrt{2\kappa_b}\:b_{\mbox{in}} \nonumber \\
\dot q&=&\omega_c p \nonumber \\
\dot p&=&-\omega_c q +G_{0a}a^{\dagger}a- G_{0b}b^{\dagger}b-\gamma_c p+\xi
\end{eqnarray}
where $\gamma_c$ is damping rate of the membrane. Also $j_{\mbox{in}}$ is input noise of field $j$ associated with cavity-input mirror interface and has only correlation function $\langle j_{\mbox{in}}(t)k_{\mbox{in}}^{\dagger}(t^{\prime}) \rangle=\delta_{jk}\delta(t-t^{\prime})$ \cite{walls2007quantum}, whereas $\xi$ is Brownian noise of the membrane and has correlation function $\langle \xi(t) \xi(t^{\prime})+\xi(t^{\prime})\xi(t)\rangle/2\approx \gamma_c(2\bar n+1)\delta(t-t^{\prime})$ in the limit of interest that is large mechanical quality of the membrane, i.e. $\omega_c/\gamma_c\gg 1$ \cite{giovannetti2001phase,benguria1981quantum}. The mean phonon number of the membrane reads $\bar n=1/(\exp{(\hbar \omega_c/k_BT)}-1)$.

The linearised Langevin equations can be obtained by splitting the operators into steady state values and fluctuating terms. In particular we write $q=q_s+\delta q$, $p=p_s+\delta p$, and $j=\alpha_{s,j}+\delta j$. By inserting these into Eq. (\ref{Eq_lgvn}) and ignoring nonlinear terms $\delta j^{\dagger}\delta j$ and $\delta j\delta q$ one gets a set of linear Langevin equations for the fluctuation of the quadratures
\begin{eqnarray}\label{Eq_llgvn}
\delta \dot x_a&=&-\kappa_a \delta x_a+\Delta_a \delta y_a+\sqrt{2\kappa_a}\:x_{\mbox{in},a} \nonumber \\
\delta \dot y_a&=&-\kappa_a \delta y_a-\Delta_a \delta x_a+G_a\delta q+\sqrt{2\kappa_a}\:y_{\mbox{in},a} \nonumber \\
\delta \dot x_b&=&-\kappa_b \delta x_b+\Delta_b \delta y_b+\sqrt{2\kappa_b}\:x_{\mbox{in},b} \nonumber \\
\delta \dot y_b&=&-\kappa_b \delta y_b-\Delta_b \delta x_b-G_b\delta q+\sqrt{2\kappa_b}\:y_{\mbox{in},b} \nonumber \\
\delta \dot q&=&\omega_c \delta p \nonumber \\
\delta \dot p&=&-\omega_c\delta q-\gamma_c \delta p+G_a\delta x_a -G_b\delta x_b+\xi 
\end{eqnarray}
where effective detuning $\Delta_a\equiv \Delta_{0a}-G_{0a}q_s$, $\Delta_b\equiv \Delta_{0b}+G_{0b}q_s$, and effective coupling $G_j\equiv \sqrt{2}G_{0j}\alpha_{s,j}$. The steady state values are given by $p_s=0$, $q_s=(G_{0a}|\alpha_{s,a}|^2-G_{0b}|\alpha_{s,b}|^2)/\omega_c$, and $\alpha_{s,j}=|E_j|/\sqrt{\kappa_j^2+\Delta_j^2}$. The quadratures of the field $x_j$ and $y_j$ are related to the field operator $j$ through $j=(x_j+iy_j)/\sqrt{2}$. This relation also applies for the input noise, i.e. $j_{\mbox{in}}=(x_{\mbox{in},j}+iy_{\mbox{in},j})/\sqrt{2}$.

For simplicity one can re-write Eq. (\ref{Eq_llgvn}) as a single matrix equation $\dot u(t)=Ku(t)+n(t)$ where the vector $u^T(t)=(\delta x_a,\delta y_a,\delta x_b,\delta y_b,\delta q,\delta p)$, $n^T(t)=(\sqrt{2\kappa_a}\: x_{\mbox{in,a}},\sqrt{2\kappa_a}\: y_{\mbox{in,a}},\sqrt{2\kappa_b}\: x_{\mbox{in,b}},\sqrt{2\kappa_b}\: y_{\mbox{in,b}},0,\xi)$, and
\begin{equation}
K=\left( \begin{array}{cccccc} -\kappa_a&\Delta_a&0 &0 &0 &0\\-\Delta_a&-\kappa_a&0 &0 &G_a &0\\0&0&-\kappa_b &\Delta_b &0 &0\\0&0&-\Delta_b &-\kappa_b &-G_b &0\\0&0&0 &0 &0&\omega_c\\G_a&0&-G_b &0 &-\omega_c &-\gamma_c \\ \end{array}\right).
\end{equation}
The solution to linearised Langevin equation is then $u(t)=M(t)u(0)+\int_0^t ds M(s)n(t-s)$ where $M(t)=\exp{(Kt)}$.

The quantum state of the fluctuations is fully characterised by covariance matrix $V_{ij}(t)\equiv \langle u_i(t)u_j(t)+u_j(t)u_i(t)\rangle/2-\langle u_i(t)\rangle \langle u_j(t)\rangle$. Note that the Gaussian nature of the initial state is maintained since we have linear dynamics and the noises involved are zero-mean Gaussian noises. One can show that the covariance matrix at time $t$ is $V(t)=M(t)V(0)M^T(t)+\int_0^t ds\; M(s)DM^T(s)$ where $D=\mbox{Diag}[\kappa_a,\kappa_a,\kappa_b,\kappa_b,0,\gamma_c(2\bar n+1)]$. A more explicit solution of the covariance matrix, after integration, is given by
\begin{eqnarray}
KV(t)+V(t)K^T&=&-D+KM(t)V(0)M^T(t) \nonumber \\
&&+M(t)V(0)M^T(t)K^T \nonumber \\
&&+M(t)DM^T(t),
\end{eqnarray}
which is linear and can easily be solved numerically. As mentioned in the main text we take the initial state to be thermal state for $c$ and coherent state for field $j$, this gives $V(0)=\mbox{Diag}[1,1,1,1,2\bar n+1,2\bar n+1]/2$. If one is only interested in steady state solution, it is guaranteed when all real parts of eigenvalues of $K$ are negative, giving $M(\infty)=0$ such that the steady state covariance matrix can be calculated from a simpler equation $KV(t_{s})+V(t_{s})K^T=-D$.

The covariance matrix $V$ describing our three-mode optomechanical system can be written in block form 
\begin{equation}
V_{abc}=\left( \begin{array}{ccc} L_{aa}&L_{ab}&L_{ac}\\ L_{ab}^T&L_{bb}&L_{bc}\\ L_{ac}^T&L_{bc}^T&L_{cc} \end{array}\right)
\end{equation}
where for $j,k=a,b,c$ the block component $L_{jk}$ is a $2\times 2$ matrix describing local mode correlation when $j=k$ and intermodal correlation when $j\ne k$. An $N$-mode covariance matrix has symplectic eigenvalues $\{\nu_k\}_{k=1}^N$ that can be computed from the spectrum of matrix $|i\Omega_N V|$ \cite{weedbrook2012gaussian} where 
\begin{equation}
 \Omega_N=\bigoplus^N_{k=1} \left( \begin{array}{cc} 0&1\\ -1 &0\end{array}\right).
 \end{equation}
For a physical covariance matrix $2 \nu_k\ge 1$ \cite{physicalV}. For an entangled system, e.g. in the partition $ab:c$, the covariance matrix will not be physical after partial transposition with respect to mode $c$ (this is equivalent to flipping the sign of the membrane's momentum fluctuation operator $\delta p$ in $V$). For our system, this unphysical $V^{T_c}$ is shown by one of its three symplectic eigenvalues $\tilde \nu_{\mbox{min}}<1/2$. Entanglement is then quantified by logarithmic negativity as follows $E_{ab:c}=\mbox{max}[0,-\ln{(2\tilde \nu_{\mbox{min}})}]$ \cite{negativity, adesso2004extremal}. Note that the separability condition, when $V^{T_c}$ has $\tilde \nu_{\mbox{min}}\ge1/2$, is sufficient and necessary for $1:N$ mode partition \cite{werner2001bound}. Entanglement $E_{a:b}$ is calculated in similar manner by only considering system $ab$ where the covariance matrix is now
\begin{equation}
V_{ab}=\left( \begin{array}{cc} L_{aa}&L_{ab}\\ L_{ab}^T&L_{bb}\\  \end{array}\right).
\end{equation}

For our calculations we vary laser power $P_j$ and laser-cavity detuning $\Delta_j$.
Other parameters have been fixed such that this setup is viable with present-day technology \cite{groblacher2009observation}.
This includes $\mu=145 \: \mbox{ng}$, $T=300\: \mbox{mK}$, $\l_j=25\: \mbox{mm}$, and $(\omega_c, \omega_{lj}, \gamma_c)=2\pi(947\times 10^3, 2.8\times10^{14}, 140)\: \mbox{Hz}$. Finesse of each cavity is $1.4\times10^4$.

The dynamics of entanglement quantified by logarithmic negativity is shown in Fig. 4 in the main text. It is clear that nonzero $E_{a:b}(t)$ implies non-classicality of the membrane.
If one is interested only in the steady state regime, Fig.~\ref{FIG_SM_SS} shows the corresponding entanglement $E_{a:b}$ while $E_{ab:c}$ is zero in this range (not shown). Note that red colour has been used in the plots for parameters that do not correspond to steady state solution.

\begin{figure}[!h]
\includegraphics[scale=0.32]{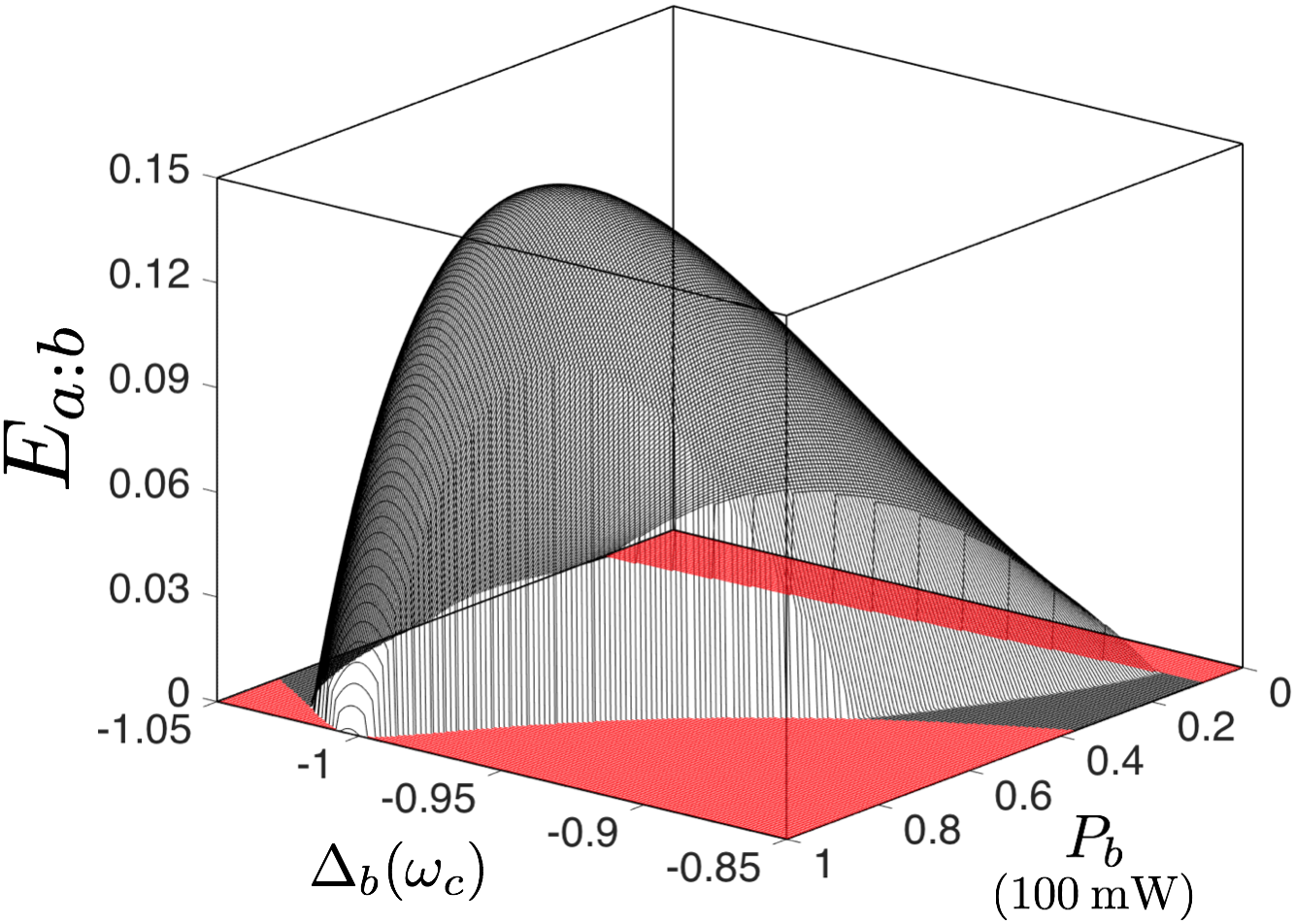}
\caption{Steady state entanglement for varying values of $P_b$, in units of $100$ mW, and $\Delta_b$, in units of the natural frequency of the membrane $\omega_c$.}
\label{FIG_SM_SS} 
\end{figure}

\end{document}